\documentclass
[
superscriptaddress,
showpacs,preprintnumbers,amsmath,amssymb]{revtex4}

\usepackage{graphicx}

\usepackage{amssymb,amsmath,amsthm}
\usepackage{epsfig}
\newtheorem{Thm}{Theorem}

\newtheorem{Prop}{Proposition}
\theoremstyle{definition}
\newtheorem{Def}{Definition}

\newtheorem{Cor}{Corollary}

\newcommand{\Z}{\mbox{$\mathbb Z$}}

\newcommand{\cI}{\mathcal{I}}

\newcommand{\Aff}{\mbox{\rm Aff}}

\newcommand{\Affq}{{\rm Aff}_q}

\newcommand{\F}{\mathbb F}

\begin{document}
\title{Quantum computational algorithm for hidden symmetry subgroup problems
on semi-direct product of cyclic groups}

\author{Jeong San Kim}\email{freddie1@suwon.ac.kr}
\affiliation{
 Department of Mathematics, University of Suwon, Kyungki-do 445-743, Korea
}
\author{Eunok Bae}\email{eobae@khu.ac.kr}
\affiliation{
 Department of Mathematics and Research Institute for Basic Sciences,
 Kyung Hee University, Seoul 130-701, Korea
}
\author{Soojoon Lee}\email{level@khu.ac.kr}
\affiliation{
 Department of Mathematics and Research Institute for Basic Sciences,
 Kyung Hee University, Seoul 130-701, Korea
}
\date{\today}

\begin{abstract}
We characterize the algebraic structure of semi-direct product of cyclic groups,
$\Z_{N}\rtimes\Z_{p}$,
where $p$ is an odd prime number
which does not divide $q-1$ for any prime factor $q$ of $N$,
and provide a polynomial-time quantum computational algorithm
solving hidden symmetry subgroup problem of the groups.
\end{abstract}
\maketitle


\section{Introduction}
Most of exponential speed-up of quantum computational algorithms
can be regarded as solving group-theoretical problems
that can be formulated within the framework of {\em hidden subgroup problem} (HSP).
Mathematically, HSP can be cast in the following terms;
given a finite group $G$ and
an oracle function (or black-box function) $f$ from $G$ to some finite set,
we say that $f$ hides a subgroup $H$ of $G$
provided that $f(a)=f(b)$ if and only if $Ha=Hb$ for all $a$ and $b$ in $G$
(that is, $a$ and $b$ belong to the same right coset of $H$),
and the task is to determine the subgroup $H$.

Whereas no classical algorithm is known to solve HSP
with polynomial query complexity
as well as with polynomial running time in the size of the group,
for every abelian group
there exists a quantum algorithm
which can efficiently solve the problem~\cite{sho97,BoLi,Kitaev},
and furthermore for an arbitrary group
there exists a quantum algorithm
which can solve the problem
with polynomial quantum query complexity~\cite{EHK04}.
In other words, HSP on an abelian group $G$ can be solved
by a quantum algorithm of running time polynomial in $\log|G|$,
and HSP on non-abelian groups can be solved
by quantum algorithms with polynomial query complexity,
although the algorithms cannot efficiently solve the problem in general.

HSP includes several algorithmically important problems;
graph isomorphism problem can be reduced to the HSP on the symmetric group
and certain lattice problems can be reduced to the HSP on the dihedral group~\cite{iso,lattice}.
For these reasons, a lot of attempts have been made
to generalize the quantum solution of the abelian HSP
to non-abelian cases~\cite{ettinger,HRT,GSVV,friedl,kuperberg,gavinsky,inui,moore,regev,BCD}.
However, the HSP for the dihedral and symmetric groups still remains unsolved.
Furthermore, the methods for solving HSP of abelian cases are known to fail for
several non-abelian groups~\cite{MRS,hmrrs}. Thus a direct generalization of quantum solutions
for abelian HSP to non-abelian cases seems rather exclusive.

Another approach toward the study of HSP is to generalize the problem itself, that is, to consider problems
dealing with more general properties of algebraic sets hidden by the oracle functions.
One of these problems is the {\em hidden symmetry subgroup problem} (HSSP)~\cite{DISW},
which can be formulated as follows;
for a group $G$ acting on some finite set $M$
and an oracle function whose level sets define a partition of $M$ by the group action,
the object we would like to recover is the group of symmetries of this partition inside $G$,
that is, the subgroup whose orbits under the action coincide with the classes of the partition,
as we will see the details in Section~\ref{Sec:Semi-Direct}.

HSP can be regarded as a special case of the HSSP
when the group acts on itself and the action corresponds to the group operation.
However, certain cases of the HSSP have exponential quantum query complexity,
in contrast to the fact that the quantum query complexity of HSP for any group is polynomial.
Thus we can say that HSSP is generally harder than HSP for some actions.
Recently, Decker, Ivanyos, Santha and Wocjan presented
an efficient quantum algorithm for HSSP on the Frobenius group,
which includes a large variety of affine groups as a special case, 
by showing that
HSSP is indeed efficiently reducible to HSP
when the action has an efficiently computable generalized base,
and that such bases for a large class of Frobenius groups
can be efficiently constructed~\cite{DISW}.

In this paper,
we first investigate algebraic properties of semi-direct product of cyclic groups,
and then construct an efficient reduction scheme of HSSP on $\Z_{N}\rtimes\Z_{p}$ to its related HSP
for the case when any prime factor $q$ of $N$
satisfies the condition that $q-1$ is not divisible by $p$.
Finally, we propose an efficient quantum algorithm
for HSSP on the group
by applying this reduction scheme to
an efficient quantum computational algorithm for the related HSP~\cite{CKL}.

This paper is organized as follows.
In Section~\ref{Sec:Semi-Direct}
we briefly review some algebraic properties and the definition of HSSP,
and in Section~\ref{A reduction scheme of HSSP to HSP}
we recall a sufficient condition of group actions proposed in~\cite{DISW},
under which a HSSP can be reduced in polynomial time to a HSP.
In Section~\ref{sec: HSSPcyclic1}
we provide some homomorphic properties of semi-direct product of cyclic groups,
and characterize its stabilizer subgroups,
and in Section~\ref{sec: HSSPcyclic2}
we show that there exist an efficient quantum algorithm
which can solve HSSP on $\Z_{N}\rtimes\Z_{p}$,
where $p$ is an odd prime number
which does not divide $q-1$ for any of the prime factors $q$ of $N$.
Finally, we summarize our result in Section~\ref{Conclusion}.

\section{Preliminaries}\label{Sec:Semi-Direct}

A group action of a group $G$ on a set $M$ is
a binary function $ \circ: G \times M \to M$ (with the notation $\circ (g,m)$ = $g \circ m$),
which satisfies $ g \circ (h \circ m) = (gh) \circ m $
and $e \circ m= m $ for any $g, h \in G$, $m \in M$ and the identity element $e$ of $G$.
We denote $g \circ L = \{g \circ m : m \in L\}$ for a subset $L \subseteq M$.

For each $m \in M$, its {\em stabilizer} subgroup $G_m$ is defined as $\{g \in G : g \circ m = m\}$,
which consists of the elements in $G$ fixing $m$ under the group action.
The group action $\circ$ of $G$ on $M$ is {\em faithful}
if $\bigcap_{m \in M} G_m = \{e\}$, that is,
$e$ is the only element of $G$ that fixes every element of $M$.
For any subgroup $H$ of $G$,
$H$ also acts naturally on $M$.
The $H$-{\em orbit} of $m \in M$
is the subset of $M$ defined as $H \circ m = \{h \circ m : h\in H\}$.
The $H$-orbits form a partition $H^* =\{ H \circ m : m \in M \}$ of $M$.
For a partition $\pi=\{\pi_1, \ldots, \pi_\ell\}$ of the set $M$,
{\em the group of symmetries} of $\pi$
is the subgroup $\pi^*   =\{ g \in G : (\forall i) \; g \circ \pi_i = \pi_i \}$,
which consists of the elements stabilizing every class of the
partition $\pi$ under the group action.

The subgroup $H^{**}$  of $G$ is the {\em closure} of $H$~\cite{Blyth},
which consists of the elements in $G$ stabilizing every $H$-orbit.
The closure of a partition $\pi$ is $\pi^{**}$,
which consists of the orbits of its group of symmetries.
We note that $H$ is always a subgroup of $H^{**}$
and $H$ is said to be {\em closed} if $H = H^{**}$, that is,
there exists a partition $\pi$ such that $H = \pi^{*}$.
Similarly, $\pi$ is said to be closed if $\pi = \pi^{**}$.
We denote by ${\cal C}(G)$ the family of all closed subgroups in $G$.

Now let us recall the formal definition of the HSSP~\cite{DISW};
for a finite group $G$, a finite set $M$,
an action $\circ: G \times M \to M$ and a family ${\cal H}$ of closed subgroups of $G$,
let us assume that an oracle function $f$ is given,
which is defined on $M$ to some finite set $S$ such that
$f(x) = f(y)$ if and only if $H \circ x = H \circ y$ for some subgroup $H \in {\cal H}$.
The HSSP is to determine the subgroup $H$.

The subsets of $M$ whose elements have the same
function value of $f$ form a partition of $M$, denoted by $\pi_f$.
Each party of this partition is called a level set of $f$.
Although there can be several subgroups of $G$ whose orbits coincide with
the level sets of $f$, the closures of these subgroups are the
same. The unique closed subgroup that satisfies the promise is
$\pi_f^*$, and this is the output of the HSSP.
($f$ is said to {\em hide $H$ by symmetries}.)


For a prime power $q$, the general affine group $\Aff_q$ is
the group of invertible affine transformations over the $\F_q$,
which can be represented as the {\em semi-direct product} of groups;
for finite groups $K$, $H$
and a homomorphism $\phi:h\mapsto \phi_h$ from $H$ to the group of automorphisms of $K$,
the semi-direct product of $K$ and $H$, denoted by $K\rtimes_\phi H$,
is the cartesian product of $K$ and $H$ with the group operation defined as
$(k,h)\cdot (k',h')=(k\cdot\phi_h(k'),h\cdot h')$.
(We use the notation $K\rtimes H$ for $K\rtimes_\phi H$
whenever $\phi$ is clear from the context.)

Using the notion of semi-direct product,
$\Aff_q$ can be represented as $\F_q \rtimes \F_q^*$,
where $\F_q^*$ denotes the multiplicative group of $\F_q$.
The natural group action of $\Affq$ on $\F_q$ is defined as $(b,a) \circ x = ax +b$.
For each $c \in \F_q$, the stabilizer of $c$ is the subgroup $H_c= \{ ((1-a)c,a) : a \in \F_q^*\}$.
$H_c$ is a closed subgroup,
and it has two orbits $\{c\}$ and $\{d \in \F_q : d \neq c \}$.

By letting ${\cal H} = \{H_c : c \in \F_q\}$,
Grover's search over $\F_q$ to find $c$ can be regarded as
a HSSP to find a closed subgroup $H_c$ in $\cal H$;
for any input $x$ and the oracle function $f_c$ such that
$f_c(x) = \delta_{c,x}$, where $\delta_{c,x}$ is the Kronecker delta,
$f_c$ hides $H_c$ as a symmetry subgroup.
Since we can recover $c$ from any generator $(b,a)$ of $H_c$
simply by computing $(1-a)^{-1}b$,
the query complexity of the HSSP is at least that of Grover's search.
Because Grover's search has query complexity $\Omega(q^{1/2})$~\cite{BBBV97},
it can be shown that
the query complexity of HSSP on the affine group $\Affq$ over
$\F_q$ is $\Omega(q^{1/2})$.

\section{A reduction scheme of HSSP to HSP}
\label{A reduction scheme of HSSP to HSP}
In this section, we recall a general condition of the group action,
under which a HSSP can be reduced in polynomial time to a HSP~\cite{DISW}.
For a given oracle function $f$ over $M$,
which hides some subgroup $H$ of $G$ by symmetries,
we construct a suitable function $f_{\rm HSP}$ over $G$, which hides $H$.

\begin{Def}
\label{Def:strongbase-def}
For a finite group $G$ and a group action $\circ: G \times M \to M$ of $G$ on the finite set $M$,
let $H$ be a subgroup of $G$, and ${\cal H}$ be a family of subgroups of $G$ including $H$.
A set $B \subseteq M$ is said to be an $H$-{\em strong base} if
\begin{equation}
\bigcap_{m \in B} H G_{g \circ m} = H,
\label{eq:strong_base_def}
\end{equation}
for every $g \in G$ and the stabilizer subgroup $G_{g \circ m}$ of $g \circ m$.
$B$ is said to be an ${\cal H}$-{\em strong base}
when it is $H$-strong for every subgroup $H \in {\cal H}$.
\end{Def}

We note that $\bigcap_{m\in M}HG_m=H^{**}$.
Thus $M$ itself is always a ${\cal C}(G)$-strong base.
Furthermore, if $B$ is an $H$-strong base,
then $B$ is also an $(x^{-1}Hx)$-strong base for every $x \in G$.
Therefore, we can show that
if ${\cal H}$ consists of conjugated subgroups
then $B$ becomes an ${\cal H}$-strong base
when it is an $H$-strong base for some $H \in {\cal H}$,
and that if $\cal H$ is closed under conjugation by elements of $G$ then
$B$ is an ${\cal H}$-{strong base} if and only if
$\bigcap_{m \in B} H G_{m} = H$ for every $H\in {\cal H}$.

Based on the concept of ${\cal H}$-strong bases,
the authors in Ref.~\cite{DISW} proposed a reduction scheme from a HSSP to a HSP.
\begin{Prop}\label{Prop:reduction}
Let $G$ be a finite group, and let $\circ$ be  an action of $G$ on $M$.
Suppose that the function $f : G \rightarrow  S$ hides some $H \in {\cal H}$ by symmetries.
Let $B = \{m_1, \ldots, m_t\}$ be an ${\cal H}$-strong base.
Then $H$ is hidden by the function $f_{\rm HSP}(g)= (f (g \circ m_1),\ldots, f(g \circ m_t))$.
\end{Prop}
Proposition~\ref{Prop:reduction} implies that
if we can choose a proper subset $B$ of $M$,
which is an ${\cal H}$-strong base,
the HSSP can always be reduced to the HSP.
Furthermore, it naturally leads us to the following proposition,
which provides a sufficient condition
for a polynomial-time reducibility of HSSP to its related HSP.
\begin{Prop}\label{prop:reduction}
Let $G$ be a finite group, $M$ a finite set, $\circ$ a polynomial time computable
action of $G$ on $M$, and ${\cal H}$ a family of subgroups of $G$.
If there exists an efficiently computable ${\cal H}$-strong base
in $M$, then HSSP on the group $G$ is polynomial time reducible to HSP of $G$.
\end{Prop}
When the group $G$ is a semi-direct product group,
an efficient characterization of $\mathcal H$-strong base~\cite{DISW} has been proposed as follows.
Let us assume that $G$ is isomorphic to the semi-direct product of its subgroups $K$ and $H$,
that is, $G \simeq K\rtimes_\phi H$.
The group action defined here is
\begin{equation}
\circ : G \times K \rightarrow K,~  g\circ x=yhxh^{-1},
\label{eq:Faction}
\end{equation}
where $x \in K$ and $g=yh$ for some $y \in K$ and $h \in H$.
If we consider the set $\mathcal H$ consisting of the all conjugate groups of $H$,
that is, $\mathcal H=\{gHg^{-1}|g \in G\}$,
then being an ${\cal H}$-strong base is equivalent to being an $H$-strong base.

For an efficient characterization of ${\cal H}$-strong base,
we recall the concept of {\em separation} among elements of $K$
with respect to the group action and its orbits;
for $u,v \in K$ with $u\neq v$, we say that an element $z\in K$ {\em separates} $u$ and $v$ if
\begin{equation}
v\circ z\not\in H\circ(u\circ z).
\label{eq:separation}
\end{equation}
Then the following proposition provides us with
a necessary and sufficient condition to characterize ${\cal H}$-strong base~\cite{DISW}.

\begin{Prop}
\label{Frobenius-base-lemma}
Let $B \subseteq K$. Then $B$ is an $H$-strong base
if and only if for all $u\neq v$ in $K$
there exists $z \in B$ which separates $u$ and $v$.
\end{Prop}

\section{Semi-direct product of cyclic groups}
\label{sec: HSSPcyclic1}
\subsection{Semi-direct product of cyclic groups and group action}

For any positive integer $M$ and $N$
and any group homomorphism $\phi$
from 
$\Z_{N}$ into 
$\Z_{M}$,
the semi-direct product group $\Z_{M}\rtimes_{\phi}\Z_{N}$ is
the set $\{(a,b): a \in \Z_{M}, b\in \Z_{N}\}$
with the group operation $(a_{1},b_{1})(a_{2},b_{2})=(a_{1}+\phi(b_{1})(a_{2}), b_{1}+b_{2})$.
For any odd prime $p$ and positive integer $n$ with $n\geq 2$,
let $G=\Z_{p^n}\rtimes_\phi\Z_p$ be the semi-direct product
with respect to a homomorphism $\phi$ from $\Z_p$ to the
automorphism group $\mathrm{Aut}\left(\Z_{p^n}\right)$ of
$\Z_{p^n}$. Because $\phi$ is a homomorphism,
we have
\begin{equation}
\phi(a)(b)=b\phi(1)(1)^a,
\label{phihomo}
\end{equation}
for any $a\in \Z_p$ and $b \in \Z_{p^n}$.
We also note that $\phi(1)(1)$ is relatively prime to $p^n$,
and thus the semi-direct product group $G$ is
completely determined by the image of $\phi(1)(1)$ in $\Z_{p^n}$.

For example, if $\phi(1)(1)=1 \in \Z_{p^n}$, then $G$ is the direct product $\Z_{p^n}\times\Z_p$.
If $\phi(1)(1)\neq 1 \pmod {p^n}$ then $p$ is the smallest positive integer
satisfying
\begin{equation}
\phi(1)(1)^p= 1\pmod {p^n},
\label{eq:order_p}
\end{equation}
that is,
$\phi(1)(1)$ is one of elements of $\Z_{p^n}^{*}$ with order $p$.
Hence, it is straightforward to verify that
$\phi(1)(1)$ is of the form
\begin{equation}
\phi(1)(1)=rp^{n-1}+1\pmod {p^n},
\label{generator}
\end{equation}
for some $r \in \{0,1,\cdots, p-1\}$~\cite{inui}.
Thus we assume that $r \neq 0$ to avoid the trivial case of the direct product $\Z_{p^n}\times\Z_p$,
and equivalently use the notions of $\phi(1)(1)$ and $rp^{n-1}+1$
for some $r \in \{1,\cdots, p-1\}$ throughout this paper.

Let us define two subgroups $K$ and $H$ of $G$ as
\begin{equation}
K=\Z_{p^n}\times \{0\},~H=\{0\}\times\Z_{p},
\label{KandH}
\end{equation}
and consider a group action $\circ : G \times K \rightarrow K$ by
\begin{eqnarray}
\left(y,h \right)\circ \left(x,0 \right)
&=&\left(y,h \right)\left(x,0 \right)\left(0,-h \right)\nonumber\\
&=&\left(y+\phi(h)(x), 0 \right),
\label{action}
\end{eqnarray}
for any  $\left(y, h\right) \in G$ and $\left(x, 0\right) \in K$
(or equivalently, for any $x, y\in \Z_{p^n}$ and $h \in \Z_p$.)

It is clear that
the identity element $\left(0,0\right)$ in $K$ is fixed by any element of $H$ under the action $\circ$.
Furthermore, the following theorem completely characterizes the elements of $K$ that are fixed by $H$.
\begin{Thm}\label{stab_H}
For any $\left(y,0\right) \in K$ and $\left(0,h\right) \in H$ satisfying $h\neq 0 \pmod {p}$,
$\left(y,0\right)$ is fixed by $\left(0,h\right)$ under the group action $\circ$
if and only if $y$ is not relatively prime to $p^n$.
\end{Thm}
\begin{proof}
From the definition of group action in Eq.~(\ref{action}), we have
\begin{eqnarray}
\left(0,h\right)\circ\left(y,0\right)
&=&\left(\phi(h)(y),0\right)\nonumber\\
&=&\left(y\phi(1)(1)^h,0\right).
\label{Hfixy}
\end{eqnarray}
If $y$ is not relatively prime to $p^n$,
then we can assume  $y=sp^j$ for some $1\leq j \leq n-1$ and $0\leq s \leq p-1$.
Together with Eq.~(\ref{generator}), we have
\begin{eqnarray}
y\phi(1)(1)^h
&=&sp^j\left(rp^{n-1}+1\right)^h \pmod {p^n}\nonumber\\
&=&sp^j\sum_{i=0}^{h}\binom{h}{i}\left(rp^{n-1}\right)^i \pmod {p^n}\nonumber\\
&=&sp^j\pmod {p^n}\nonumber\\
&=&y\pmod {p^n},
\label{phi11bi}
\end{eqnarray}
and thus $\left(0,h\right)\circ\left(y,0\right)=\left(y,0\right)$ for any $\left(0,h\right) \in H$ if $y$
is not relatively prime to $p^n$.

Conversely, let us suppose that $\left(0,h\right)$ fixes $\left(y,0\right)$
with $y$ being relatively prime to $p^n$.
From Eq.~(\ref{Hfixy}), we have
\begin{equation}
y\phi(1)(1)^h=y \pmod {p^n},
\label{Hfixy2}
\end{equation}
or equivalently, $y\left(\phi(1)(1)^h-1\right)$ is divided by $p^n$.
Because $y$ is relatively prime to $p^n$,
Eq.~(\ref{Hfixy2}) is true if and only if
\begin{equation}
\phi(1)(1)^h=1 \pmod {p^n}.
\label{Hfixy3}
\end{equation}
In other words, $\left(0,h\right)$ fixes $\left(y,0\right)$ with respect to the action $\circ$
if and only if Eq.~(\ref{Hfixy3}) holds.
However, this contradicts to the fact that $p$ is the smallest positive integer satisfying Eq.~(\ref{eq:order_p})
since $1 \leq h \leq p-1$.
Thus for any $\left(0,h\right) \in H$ and $\left(y,0\right) \in K$ such that $y$ is relatively prime to $p^n$,
$\left(0,h\right)$ does not fix $\left(y,0\right) \in K$.
\end{proof}

Let us define the subset $P_0\times \{0\}$ of $K$
where $P_0=\{pk| 0 \leq k \leq p^{n-1}-1\}$ consists of the elements in $\Z_{p^n}$,
which are not relatively prime to $p^n$.
Then Theorem~\ref{stab_H} implies that $H$ is the stabilizer subgroup of $G$
that fixes every element in $P_0\times \{0\}$.
For this reason, we also denote $H=H_{P_0 \times\{0\}}$.
Theorem~\ref{stab_H} also implies that
the semi-direct product of cyclic groups $G=\Z_{p^n}\rtimes_\phi\Z_p$
for general $p$ and $n$ under the action in Eq.~(\ref{action}) is not a Frobenius group
because, not only the identity element $\left(0,0\right)$,
every element in $H$ has more than one fixed element.

The following theorem shows that the action of $H$ on any element of $K$ that is not in $P_0\times \{0\}$
is {\em faithful}, that is,
for any $\left(y,0\right) \in K$ such that $y$ is relatively prime to $p^n$,
two different elements of $H$ lead $\left(y,0\right)$ to different elements in $K$ under the action $\circ$.
\begin{Thm}
For $\left(y,0\right) \in K$ such that $y$ is relatively prime to $p^n$,
$\left(0,h\right)\circ \left(y,0\right)\neq \left(0,h'\right)\circ \left(y,0\right)$
for any $h, h' \in \Z_p$ such that $h \neq h' \pmod {p}$.
\label{actdiff}
\end{Thm}

\begin{proof}
Suppose $\left(0,h\right)\circ \left(y,0\right)=\left(0,h'\right)\circ \left(y,0\right)$.
Because $y$ is relatively prime to $p^n$, let $y=pk+t$ for some $t \in \{1,2,\cdots,p-1\}$, then
\begin{eqnarray}
\left(0,h\right)\circ \left(y,0\right)&=&\left(\phi(h)(y),0\right)=\left(y\phi(1)(1)^h,0\right),\nonumber\\
\left(0,h'\right)\circ \left(y,0\right)&=&\left(\phi(h')(y),0\right)=\left(y\phi(1)(1)^{h'},0\right).
\end{eqnarray}
By the assumption,
we have $y\phi(1)(1)^h=y\phi(1)(1)^{h'}\pmod {p^n}$,
which is equivalent to
\begin{equation}
y\left(\phi(1)(1)^h-\phi(1)(1)^{h'}\right)=0\pmod {p^n}.
\label{eq: actdiff2}
\end{equation}
Since $y$ is not a zero divisor in $\Z_{p^n}$, we have
\begin{equation}
\phi(1)(1)^h=\phi(1)(1)^{h'}\pmod {p^n}.
\label{eq: actdiff3}
\end{equation}
However Eq.~(\ref{eq: actdiff3}) implies $\phi(1)(1)^{h-h'}=1\pmod {p^n}$,
which contradicts to the fact in Eq.~(\ref{eq:order_p})
stating that $p$ is the smallest integer
satisfying $\phi(1)(1)^{p}=1\pmod {p^n}$
because $0<h-h'<p$ (without loss of generality, we may assume $h>h'$).
Thus $\left(0,h\right)\circ \left(y,0\right)\neq\left(0,h'\right)\circ \left(y,0\right)$.
\end{proof}

From Theorem~\ref{stab_H} together with Theorem~\ref{actdiff},
we note that the orbits of $H$ are singleton subsets $\{\left(pk, 0\right)\}$ of $P_0\times \{0\}$
and some subsets of $K$, each consisting of $|H|$ number of elements.
The theorems also implies that $H$ is a closed subgroup and its orbits form a closed partition of $K$.
In the following subsection,
we will consider the general form of closed subgroups of $G=\Z_{p^n}\rtimes_\phi\Z_p$
and their orbits in accordance of $H$.

\subsection{Stabilizer Subgroups}
\label{Stabilizer Subgroups}
In this section,
we consider stabilizer subgroups of each element in $K$
with respect to the group action in Eq.~(\ref{action}).
Let us first consider a partition of $K$; for each $t \in \{0,1,\cdots, p-1\}$,
we define $P_t$ to be the set of elements in $\Z_{p^n}$
whose remainder is $t$ when divided by $p$,
that is, $P_t=\{pk+t| 0 \leq k \leq p^{n-1}-1\}$.
It is clear that $K$ can be partitioned into subsets $P_t\times \{0\}$.

\begin{Thm}
For any $\left(x,0\right) \in P_t\times \{0\}$ with $t \in \{0,1,\cdots, p-1\}$,
\begin{equation}
\left(x,0\right)H\left(-x,0\right)=\left(t,0\right)H\left(-t,0\right)
\label{H_t}
\end{equation}
where $\left(x,0\right)H\left(-x,0\right)=
\{\left(x,0\right)\left(0,h\right)\left(-x,0\right)| \left(0,h\right)\in H\}$
is the conjugate subgroup of $H$ in $G$.
\label{H_P_t}
\end{Thm}

\begin{proof}
Because $x\in P_t$, $x=pk+t$ for some $0 \leq k \leq p^{n-1}-1$,
\begin{eqnarray}
x-\phi(h)(x)
&=& pk+t-\phi(h)\left(pk+t\right)\nonumber\\
&=& pk+t-\phi(h)\left(pk\right)-\phi(h)\left(t\right)\nonumber\\
&=& t-\phi(h)(t),
\label{eq: H_t2}
\end{eqnarray}
where the last equality is due to
\begin{eqnarray}
\phi(h)\left(pk\right)
&=&pk\phi(1)(1)^h \pmod {p^n}\nonumber\\
&=&pk\left(rp^{n-1}+1\right)^h \pmod {p^n}\nonumber\\
&=&pk\sum_{i=0}^{h}\binom{h}{i}\left(rp^{n-1}\right)^i \pmod {p^n}\nonumber\\
&=&pk \pmod {p^n}.
\label{eq: H_t3}
\end{eqnarray}

Now for any $\left(0,h\right)\in H$, we have
\begin{eqnarray}
\left(x,0\right)\left(0,h\right)\left(-x,0\right)
&=&\left(x-\phi(h)(x),0 \right)\nonumber\\
&=&\left(t-\phi(h)(t),0\right)\nonumber\\
&=&\left(t,0\right)\left(0,h\right)\left(-t,0\right),
\label{eq: H_t}
\end{eqnarray}
which completes the proof.
\end{proof}

Now we have the following theorem,
which completely characterizes the stabilizer subgroups of each element in $K$.
\begin{Thm}
For any $\left(y,0\right) \in K$ such that $y\in P_t$,
$\left(y,0\right)$ is fixed by $\left(x,h\right) \in G$ under the group action
$\circ$ if and only if  $\left(x,h\right) \in \left(t,0\right)H\left(-t,0\right)$ .
\label{H_tfix}
\end{Thm}

\begin{proof}
Because $y\in P_t$, let $y=pk+t$ for some $k \in \{0,1,\cdots,p^{n-1}-1\}$,
then for any $\left(0,h\right) \in H$ we have,
\begin{eqnarray}
\left(t,0\right)\left(0,h\right)\left(-t,0\right)\circ\left(y,0\right)
&=&\left(t-\phi(h)(t),h \right)\circ\left(y,0\right)\nonumber\\
&=&\left(t-\phi(h)(t)+\phi(h)(y),0 \right),
\label{eq:tht}
\end{eqnarray}
with
\begin{eqnarray}
t-\phi(h)(t)+\phi(h)(y)
&=&t-\phi(h)(t)+\phi(h)(pk+t)\pmod {p^n}\nonumber\\
&=&t-\phi(h)(t)+\phi(h)(pk)+\phi(h)(t)\pmod {p^n}\nonumber\\
&=&t+\phi(h)(pk)\pmod {p^n}\nonumber\\
&=&t+pk\phi(1)(1)^h\pmod {p^n}\nonumber\\
&=&t+pk\left(rp^{n-1}+1\right)^h\pmod {p^n}\nonumber\\
&=&t+pk\pmod {p^n}
\label{eq:tht2}
\end{eqnarray}
where the last equality is 
due to the binomial expansion of $\left(rp^{n-1}+1\right)^h$ under modulo $p^n$.
Now we have
\begin{eqnarray}
\left(t,0\right)\left(0,h\right)\left(-t,0\right)\circ\left(y,0\right)
&=&\left(t-\phi(h)(t)+\phi(h)(y),0 \right)\nonumber\\
&=&\left(pk+t,0\right)\nonumber\\
&=&\left(y,0\right),
\label{eq:tht3}
\end{eqnarray}
which implies that any element in $\left(t,0\right)H\left(-t,0\right)$ fixes $\left(y,0\right)$.

Conversely, suppose that there exists $\left(x,h\right) \in G$
which fixes $\left(y,0\right)$ under the action $\circ$, that is
\begin{equation}
\left(x,h\right)\circ \left(y,0\right)=\left(y,0\right),
\label{eq:tht4}
\end{equation}
where
\begin{eqnarray}
\left(x,h\right)\circ \left(y,0\right)
&=&\left(x+\phi(h)(y),0\right)\nonumber\\
&=&\left(x+\phi(h)(pk+t),0\right)\nonumber\\
&=&\left(x+\phi(h)(pk)+\phi(h)(t),0\right).
\label{eq:tht5}
\end{eqnarray}
From Eq.~(\ref{eq:tht4}), we have
\begin{eqnarray}
y&=&x+\phi(h)(pk)+\phi(h)(t)\pmod {p^n}\nonumber\\
&=&x+pk\left(\phi(1)(1)\right)^h+\phi(h)(t)\pmod {p^n}\nonumber\\
&=&x+pk+\phi(h)(t)\pmod {p^n}
\label{eq:tht6}
\end{eqnarray}
where $y=pk+t$.
Thus $x=t-\phi(h)(t)\pmod {p^n}$, or equivalently
\begin{eqnarray}
\left(x,h\right)
&=&\left(t-\phi(h)(t),h\right)\nonumber\\
&=&\left(t,0\right)\left(0,h\right)\left(-t,0\right)
\in \left(t,0\right)H\left(-t,0\right),
\label{eq:tht7}
\end{eqnarray}
which completes the proof.
\end{proof}

From Theorem~\ref{H_P_t} and Theorem~\ref{H_tfix},
we note that, for each $t \in \{0,1,\cdots, p-1\}$,
the conjugate group $\left(t,0\right)H\left(-t,0\right)$ of $H$
is the stabilizer group of each elements in $P_t\times \{0\}$
with respect to the group action $\circ$.
Similarly with Theorem~\ref{actdiff},
it is also straightforward to verify
that $\left(t,0\right)H\left(-t,0\right)$ acts faithfully on any element of $K$ that is not in $P_t\times \{0\}$.
They are closed subgroups of $G$ and their orbits form closed partitions.
We will denote $\mathcal H$ the set of all conjugate subgroups of $H$ in $G$;
\begin{equation}
\mathcal H=\{\left(t,0\right)H\left(-t,0\right)| 0\leq t \leq p-1 \}.
\label{calH}
\end{equation}

\section{Quantum Algorithm for HSSP on $\Z_{p^n}\rtimes_\phi\Z_p$}
\label{sec: HSSPcyclic2}
In this section,
we present an efficient quantum algorithm for HSSP defined on $G=\Z_{p^n}\rtimes_\phi\Z_p$
with respect to the group action in Eq.~(\ref{action})
and the set of closed subgroups $\mathcal H$ in Eq.~(\ref{calH}).
By considering an efficient reduction scheme of HSSP defined on $G$ onto its related HSP,
we show that there exists a quantum algorithm solving HSSP on $\Z_{p^n}\rtimes_\phi\Z_p$
in a polynomial time with respect to the size of the group.

From Propositions~\ref{Prop:reduction} and \ref{prop:reduction},
we note that for a given set $G$ with a set of closed subsets $\mathcal H$,
there exists a polynomial-time reduction scheme from HSSP to HSP
if we can efficiently construct an $\mathcal H$-strong base of small size.
For the case when $G$ is a semi-direct product group,
Proposition~\ref{Frobenius-base-lemma} provides us
with an efficient way to convince the existence of an $\mathcal H$-strong base.
Furthermore, if the group is a semi-direct product of cyclic groups, $G=\Z_{p^n}\rtimes_\phi\Z_p$,
the following theorem gives a lower bound of the probability
that element in $K$ separates given two distinct elements of $K$
with respect to the action in Eq.~(\ref{action}).

\begin{Thm}
For given $\left(u,0\right)$ and $\left(v,0\right)$ in $K$ with $u\neq v \pmod {p^n}$
and a randomly chosen element $\left(z,0\right)$ from $K$,
the probability that $\left(z,0\right)$ separates $\left(u,0\right)$ and $\left(v,0\right)$
is no less than $1-\frac{\left(p-1\right)^2}{p\left(p^n-1 \right)}$.
\label{thm:uvsepa}
\end{Thm}

\begin{proof}
Let us suppose that $\left(z,0\right)$ does not separate $\left(u,0\right)$ and $\left(v,0\right)$.
From the definition of separation in (\ref{eq:separation}),
we have $\left(v,0\right)\circ\left(z,0\right) \in H\circ[\left(u,0\right)\circ\left(z,0\right)]$.
In other words, there exists an element $\left(0,h\right)$ in $H$ such that
\begin{equation}
\left(v,0\right)\circ\left(z,0\right)=\left(0,h\right)\circ[\left(u,0\right)\circ\left(z,0\right)],
\label{eq:notsep}
\end{equation}
which is equivalent to $\left(v+z,0\right)=\left(\phi(h)(u+z),0\right)$
by the definition of group action in Eq.~(\ref{action}).
Thus $\left(z,0\right)$ does not separate $\left(u,0\right)$ and $\left(v,0\right)$ if and only if
there exists $\left(0,h\right) \in H$ (or equivalently there exists $h \in \Z_p$) such that
\begin{equation}
v+z=\phi(h)(u+z)\pmod {p^n}.
\label{eq:notsep2}
\end{equation}
Because $\phi(h)$ is a homomorphism, Eq.~(\ref{eq:notsep2}) is also equivalent to
\begin{equation}
v-\phi(h)(u)=\phi(h)(z)-z\pmod {p^n}.
\label{eq:notsep3}
\end{equation}

Now we note that the right-hand side of Eq.~(\ref{eq:notsep3}) becomes
\begin{eqnarray}
\phi(h)(z)-z
&=&z\left(\phi(1)(1)^h-1\right)\pmod {p^n}\nonumber\\
&=&z[\left(rp^{n-1}+1\right)^h-1]\pmod {p^n}\nonumber\\
&=&zhrp^{n-1}\pmod{p^n},
\label{eq:notsep4}
\end{eqnarray}
where the last equality is by the binomial expansion of $\phi(1)(1)^h=\left(rp^{n-1}+1\right)^h$.
Similarly, the left-hand side of Eq.~(\ref{eq:notsep3}) can also be expressed as
\begin{eqnarray}
v-\phi(h)(u)
&=&v-u\phi(1)(1)^h\pmod {p^n}\nonumber\\
&=&v-u\left(rp^{n-1}+1\right)^h\pmod {p^n}\nonumber\\
&=&v-u-uhrp^{n-1}\pmod{p^n}.
\label{eq:notsep5}
\end{eqnarray}

From Eq.~(\ref{eq:notsep3}) together with Eq.~(\ref{eq:notsep5}) and Eq.~(\ref{eq:notsep4}),
we note that
$\left(z,0\right)$ does not separate $\left(u,0\right)$ and $\left(v,0\right)$ if and only if
there exists $\left(0,h\right) \in H$ such that
\begin{equation}
v-u=(z+u)hrp^{n-1}\pmod {p^n}.
\label{eq:notsep6}
\end{equation}

Case 1:  Let us first consider the cases when $v-u$ is not divisible by $p^{n-1}$, that is
\begin{equation}
v-u=cp^{n-1}+d \pmod {p^n},
\label{vnotu}
\end{equation}
for some $c \in \{0,1,\cdots,p-1\}$ and $d \in \{1,\cdots,p-1\}$.
For this case, it is readily seen that
Eq.~(\ref{eq:notsep6}) never holds because
\begin{eqnarray}
v-u-(z+u)hrp^{n-1}
&=&cp^{n-1}+d-(z+u)hrp^{n-1} \pmod {p^n}\nonumber\\
&=&[c-(z+u)hr]p^{n-1}+d \pmod {p^n}\nonumber\\
&\neq&0\pmod {p^n},
\label{vnotu2}
\end{eqnarray}
for any $\left(0,h\right) \in H$,
and thus every $\left(z,0\right)$ in $K$ separates $\left(u,0\right)$ and $\left(v,0\right)$.

Case 2: Now let us consider the cases when $v-u$ is divisible by $p^{n-1}$, that is,
\begin{equation}
v-u=cp^{n-1}\pmod {p^n},
\label{vnotu}
\end{equation}
for some $c \in \{1,\cdots,p-1\}$ (because $u \neq v$, $c\neq0$).
For this case, Eq.~(\ref{eq:notsep6}) becomes
\begin{equation}
cp^{n-1}=(z+u)hrp^{n-1}\pmod {p^n},
\label{eq:notsep7}
\end{equation}
which is equivalent to
\begin{equation}
zp^{n-1}=(ch^{-1}r^{-1}-u)p^{n-1}\pmod {p^n},
\label{eq:notsep8}
\end{equation}
for some $h \in \Z_p$.

Here we note that $h \in \{1, 2,\cdots p-1\}$ because $u \neq v \pmod {p^n}$,
and also $r \in \{1, 2,\cdots p-1\}$ because $\phi(1)(1)=rp^{n-1}+1\neq 1 \pmod {p^n}$.
In other words, neither $h$ nor $r$ is a zero divisor in $\Z_{p^n}$,
and thus their inverse elements also exist in $\Z_{p^n}$.
Furthermore, Eq.~(\ref{eq:notsep8}) holds if and only if $z=ch^{-1}r^{-1}-u \pmod {p}$,
that is,
\begin{equation}
z=ch^{-1}r^{-1}-u+mp,
\label{eq:notsep9}
\end{equation}
for some $m \in \{0, 1,\cdots, p^{n-1}-1\}$.

Eq.~(\ref{eq:notsep9}) implies that
given $u$ and $v$ satisfying Eq.~(\ref{vnotu}),
there are $p^{n-1}$ possible choices of $m$
for each $h \in \{1, 2,\cdots p-1\}$  such that Eq.~(\ref{eq:notsep9}) holds.
In other words, if $u$ and $v$ satisfy Eq.~(\ref{vnotu})
then there are $(p-1)p^{n-1}$ choices of $(z,0)$ in $K$,
for which $(z,0)$ does not separate $\left(u,0\right)$ and $\left(v,0\right)$.
For this case, the number of $(z, 0)$ in $K$ separating $\left(u,0\right)$ and $\left(v,0\right)$
is $p^{n}-(p-1)p^{n-1}=p^{n-1}$,
which is the number of $z$ in $\Z_{p^n}$ that does not satisfy Eq.~(\ref{eq:notsep9}).

Now let us consider the probability of randomly chosen $(z, 0)$ in $K$
that separates $\left(u,0\right)$ and $\left(v,0\right)$.
From Case 1 and 2,
we note that every $(z, 0)$ in $K$ separates $\left(u,0\right)$ and $\left(v,0\right)$
if $v-u$ is not divisible by $p^{n-1}$.
If $v-u$ is divisible by $p^{n-1}$
then there are $p^{n-1}$ number of $(z, 0)$ separating $\left(u,0\right)$ and $\left(v,0\right)$.
Thus the probability of randomly chosen $(z, 0)$ in $K$ that separates $\left(u,0\right)$ and $\left(v,0\right)$
is
\begin{align}
Prob[(z,0)&~\mathrm{separates}~(u,0)~\mathrm{and}~(v,0)]\nonumber\\
&=Prob[u\neq v \pmod {p^{n-1}}]\cdot 1 +Prob[u= v \pmod {p^{n-1}}]\cdot\frac{1}{p}.
\label{eq:prob1}
\end{align}

If $u= v \pmod {p^{n-1}}$, Eq.~(\ref{vnotu}) implies that for every $u$ in $\Z_{p^n}$,
there are $p-1$ number of possible $v$ satisfying $u=v \pmod {p^{n-1}}$.
Thus the total number of the unordered pairs $\{u, v\}$ satisfying $u=v \pmod {p^{n-1}}$ is $p^n(p-1)/2$
(the factor $1/2$ is to avoid doubly counting the unordered pair $\{u, v\}$).
Because there are $\binom{p^n}{2}$ ways to choose  $\{u, v\}$ from $\Z_{p^n}$,
we have
\begin{eqnarray}
Prob[u= v \pmod {p^{n-1}}]&=&\frac{p^n(p-1)/2}{\binom{p^n}{2}}=\frac{p-1}{p^n-1},\nonumber\\
Prob[u\neq v \pmod {p^{n-1}}]&=&1-\frac{p-1}{p^n-1},
\label{eq:prob2}
\end{eqnarray}
and together with Eq.~(\ref{eq:prob1}),
we have
\begin{align}
Prob[(z,0)~\mathrm{separates}~(u,0)~\mathrm{and}~(v,0)]
&=\left(1-\frac{p-1}{p^n-1}\right)\cdot 1 +\frac{p-1}{p^n-1}\cdot\frac{1}{p}\nonumber\\
&=1-\frac{\left(p-1\right)^2}{p\left(p^n-1 \right)}.
\label{eq:prob3}
\end{align}
\end{proof}

Theorem~\ref{thm:uvsepa} implies that
a randomly chosen element $(z,0)$ from $K$ separates given $(u,0)$ and $(v,0)$ with large probability.
In other words,
the probability that a randomly chosen element $(z,0)$ from $K$ does not separate given $(u,0)$ and $(v,0)$
is exponentially small with respect to the logarithm of the size of the group,
when the group is a semi-direct product of cyclic groups.
This idea leads us to the following theorem,
which assures the existence of an $\mathcal H$-strong base of small size with high probability
for this semi-direct product of cyclic groups.

\begin{Thm}
\label{The: semi-base}
Let $G=\Z_{p^n}\rtimes_\phi\Z_p$ be the semi-direct product of cyclic groups
with an odd prime $p$ and a positive integer $n$ such that $n\geq 2$.
$K=\Z_{p^n}\times \{0\}$ and $H=\{0\}\times\Z_{p}$ are two subgroups of $G$
where $G$ acts on $K$ with respect to the group action in Eq.~(\ref{action}),
and $\mathcal H$ is the set of all conjugate groups of $H$ in $G$.
If $B \subseteq K$ is a uniformly random set of size $\ell$,
with $\ell = \Theta\left(\frac{\ln |K| \log 1/ \epsilon}{\ln\left(\frac{p^n-1}{p-1}\right)}\right)$,
then $B$ is an $\mathcal H$-strong base with probability of at least $1-\epsilon$.
\end{Thm}

\begin{proof}
Let $B$ be a uniformly random subset of $K$ of size $\ell$.
By Proposition~\ref{Frobenius-base-lemma},
it is sufficient to prove that for every $u \neq v \pmod {p^n}$,
there exists an element in $B$ which separates $\left(u,0\right)$ and $\left(v,0\right)$
with probability of at least $1-\epsilon$.
In this proof, we will consider an upper bound of the probability of the opposite event.

From Theorem~\ref{thm:uvsepa},
the probability that a random $(z, 0)$ from $K$ does not separate
$\left(u,0\right)$ and $\left(v,0\right)$ for a fixed pair $u \neq v \pmod {p^n}$
is at most $\frac{p-1}{p^n-1}$.
Therefore, the probability that none of the elements in $B$ separates $\left(u,0\right)$ and $\left(v,0\right)$ is
not more than $\left(\frac{p-1}{p^n-1}\right)^{\ell}$.
Thus, the probability that for some pair $u \neq v \pmod {p^n}$
none of the elements in $B$ separates $u$ and $v$ is
less than or equal to $\binom{|K|}{2}\left(\frac{p-1}{p^n-1}\right)^{\ell}$,
which is at most $\epsilon$ by the choice of $\ell$.
\end{proof}

For $G=\Z_{p^n}\rtimes_\phi\Z_p$,
Theorem~\ref{The: semi-base} implies that
we can efficiently compute an $\mathcal H$-strong base of small size for the set of closed subgroups $\mathcal H$.
Therefore, by Proposition~\ref{prop:reduction},
HSSP on $\Z_{p^n}\rtimes_\phi\Z_p$ is efficiently reduced to a HSP on $\Z_{p^n}\rtimes_\phi\Z_p$.
Finally, we would like to remark that
there exists a polynomial-time quantum algorithm solving HSP on $\Z_{p^n}\rtimes_\phi\Z_p$
for any odd prime $p$ and positive integer $n$~\cite{inui,CKL}.
Thus we can have an efficient quantum algorithm for HSSP on $\Z_{p^n}\rtimes_\phi\Z_p$.

\begin{Cor}
\label{The: QAHSSP1}
Let $G=\Z_{p^n}\rtimes_\phi\Z_p$ be the semi-direct product of cyclic groups with an odd
prime $p$ and a positive integer $n$ such that $n\geq 2$. $K=\Z_{p^n}\times \{0\}$,
$H=\{0\}\times\Z_{p}$ are two subgroups of $G$ where $G$ acts on $K$ with respect to the group action
in Eq.~(\ref{action}) and $\mathcal H$ is the set of all conjugate groups of $H$ in $G$.
Then there exists a polynomial-time quantum algorithm solving HSSP on $G$.
\end{Cor}

Now, we consider a possible reduction scheme of HSSP defined on $\Z_{N}\rtimes_\phi\Z_p$
to a HSSP on $\Z_{p^n}\rtimes_\phi\Z_p$ for some case of $N$,
by using the same arguments as in Ref.~\cite{CKL}.
We first consider the case when $N=q^sp^n$ for some prime $q$ such that
$(p,q)=1$ and $p$ does not divide $q-1$,
and we further consider more general case of $N$.

If $N=q^sp^n$,
the fundamental theorem of finitely generated abelian groups implies
that $\Z_{N}$ is isomorphic to $\Z_{q^{s}}\times\Z_{p^{n}}$,
and thus we will assume $G=(\Z_{q^{s}}\times\Z_{p^{n}})\rtimes_\phi\Z_p$.
Similar to the case of HSSP on $\Z_{p^n}\rtimes_\phi\Z_p$,
let us consider the subgroups of $G$,
$K=(\Z_{q^{s}}\times\Z_{p^{n}})\times \{0\}$ and $H=\{(0,0) \}\times\Z_p$,
and the group action $\circ : G \times K \rightarrow K$ defined by
\begin{eqnarray}
\left(a,b,h \right)\circ \left(x,y,0 \right)
&=&\left(a,b,h \right)\left(x,y,0 \right)\left(0,0,-h \right)\nonumber\\
&=&\left((a,b)+\phi(h)(x,y), 0 \right),
\label{action2}
\end{eqnarray}
for any  $\left(a,b,h\right) \in G$  and $\left(x,y,0\right) \in K$.
(or equivalently, for any $a, x\in \Z_{q^s}$, $b, y\in \Z_{p^n}$ and $h \in \Z_p$.)

The set of closed subgroup $\mathcal H'$
is given by the set of all conjugate groups of $H$,
and the oracle function $f$ is defined on $K$ to some finite set $S$
such that
\begin{equation}
f(x,y) = f(x',y') \Longleftrightarrow H' \circ (x,y,0) = H' \circ (x',y',0),
\label{f2}
\end{equation}
for some subgroup $H' \in {\cal H'}$.
The task of HSSP on $G$ is to determine the subgroup $H'$.

We now take into account the following proposition~\cite{CKL}
\begin{Prop}\label{Lem:note2}
Let $p$ and $q$ be distinct primes satisfying \mbox{$p\nmid q-1$},
then
\begin{equation}
(\Z_{q^{s}}\times\Z_{p^{n}})\rtimes_{\phi}\Z_{p}\cong
\Z_{q^{s}}\times (\Z_{p^{n}}\rtimes_{\psi}\Z_{p})
\label{eq:Lem2}
\end{equation}
for some homomorphism $\psi$ from $\Z_{p}$ to $\mathrm{Aut}(\Z_{p^{n}})$.
\end{Prop}
\begin{proof}
Since $\phi(p)=\phi(0)$ is the identity map $\cI$ on $\Z_{q^{s}}\times \Z_{p^{n}}$,
we have
\begin{equation}
(1,0)=\cI(1,0)= \phi(p)(1,0)=\phi(1)^{p}(1,0)=(a^{p},0),
\label{eq:phi1p}
\end{equation}
where $(a,0)=\phi(1)(1,0)$ and $a^{p}= 1\pmod{q^{s}}$.
Since the order of $\Z_{q^{s}}^{*}$ is $q^{s-1}(q-1)$ and $p\nmid q-1$,
we obtain that $a$ must be $1$, that is, $\phi$ trivially acts on $\Z_{q^{s}}$.
Thus, for each $\alpha\in\Z_p$,
$\phi(\alpha)=\cI_0\times \psi(\alpha)$,
where $\cI_0$ is the identity map on $\Z_{q^{s}}$ and
$\psi$ is a homomorphism from $\Z_{p}$ to $\mathrm{Aut}(\Z_{p^{n}})$.

Therefore, the operation of the semi-direct product group is as follows:
\begin{eqnarray}
((a,b),c)((a',b'),c')&=&((a,b)+\phi(c)(a',b'),c+c')\nonumber\\
&=&(a+a',b+\psi(c)(b'),c+c'),\nonumber\\
\label{eq:operation_sdpg}
\end{eqnarray}
which implies Eq.~(\ref{eq:Lem2}).
\end{proof}

Proposition~\ref{Lem:note2} implies that for
any $(a,b,0) \in K$, and $(0,0,h) \in H$, we have
\begin{eqnarray}
(a,b,0)(0,0,h)(-a,-b,0)
&=&(a,b,h)(-a,-b,0)\nonumber\\
&=&\left((a,b)+\phi(h)(-a,-b), h\right)\nonumber\\
&=&\left(0,b-\psi(h)(b),h\right).
\label{conj1}
\end{eqnarray}
Thus, the set of closed subgroups $\mathcal H'$ consists of all conjugate groups of $H$
whose element has $0$ in the first coordinate;
\begin{equation}
\mathcal H'=\{(0,b,0)H(0,-b,0)|b \in \Z_{p^{n}}\}.
\label{conjH2}
\end{equation}
From Theorem~\ref{H_P_t} in Section~\ref{Stabilizer Subgroups},
we have the following corollary.
\begin{Cor}
For any $\left(0,b,0\right) \in K$ such that $b\in P_t=\{pk+t| 0 \leq k \leq p^{n-1}-1\}$,
\begin{equation}
\left(0,b,0\right)H\left(0,-b,0\right)=\left(0,t,0\right)H\left(0,-t,0\right).
\label{H_t2}
\end{equation}
\label{H_P_t2}
\end{Cor}

In other words, if we recall Eq.~(\ref{calH}), which is the set $\mathcal H$ of the closed subgroups defined for the HSSP on $\Z_{p^n}\rtimes_\phi\Z_p$,
we note that there exists a natural one-to-one correspondence between $\mathcal H'=\{\left(0,t,0\right)H\left(0,-t,0\right)| 0\leq t \leq p-1 \}$ and
$\mathcal H=\{\left(t,0\right)H\left(-t,0\right)| 0\leq t \leq p-1 \}$.

Now we characterize the group action in Eq.~(\ref{action2}) and the stabilizer subgroups of each element in $K$ under this action.
For any $(a,b,h) \in G$ and $(x,y,0) \in K$, suppose $(x,y,0)$ is fixed by $(a,b,h)$ under the action, then we have
\begin{eqnarray}
(a,b,h)\circ(x,y,0)
&=&(a,b,h)(x,y,0)(0,0,-h)\nonumber\\
&=&\left((a,b)+\phi(h)(x,y),0\right)\nonumber\\
&=&\left(a+x,b+\psi(h)(y),0\right)\nonumber\\
&=&(x,y,0).
\label{fix2}
\end{eqnarray}
Thus for any $(a,b,h) \in G$, if $(a,b,h)$ fixes any element $(x,y,0)$ in $K$ then  $a+x=x \pmod {q^s}$,
which implies $a=0 \pmod {q^s}$.
In other words,  $(a,b,h)=(0,b,h)$ belongs to a conjugate group of $H$ in $\mathcal H'$.
We also note that Eq.~(\ref{fix2}) implies that $b+\phi(h)(y)=y \pmod {p^n}$.
Thus we have the following corollary.

\begin{Cor}
For any $\left(x,y,0\right) \in K$ such that $y\in P_t=\{pk+t| 0 \leq k \leq p^{n-1}-1\}$,
$\left(x,y,0\right)$ is fixed by $\left(0,b,h\right) \in G$ under the group action $\circ$
if and only if  $\left(0,b,h\right) \in \left(0,t,0\right)H\left(0,-t,0\right)$.
\label{H_tfix2}
\end{Cor}
\begin{proof}
This is a direct consequence from Theorem~\ref{H_tfix} in Section~\ref{Stabilizer Subgroups}.
\end{proof}

From the definition of the oracle function in Eq.~(\ref{f2}),
we note that for any $(x,y,0)$ and $(x',y',0)$ in $K$,
we have $f(x,y) = f(x',y')$  if and only if
\begin{equation}
H' \circ (x,y,0) = H' \circ (x',y',0),
\label{f3}
\end{equation}
for some $H' \in \mathcal H'$.
By Corollary~\ref{H_P_t2}, we also note that
$H'=\left(0,t,0\right)H\left(0,-t,0\right)$ for some $t\in \{0,1,\cdots, p-1\}$.
Thus Eq.~(\ref{f3}) is equivalent to the existence of some $(0,0,h)$ and $(0,0,h')$ in $H$ such that
\begin{eqnarray}
\left[(0,t,0)(0,0,h)(0,-t,0)\right]\circ (x,y,0)= \left[(0,t,0)(0,0,h')(0,-t,0)\right]\circ (x',y',0),
\label{f4}
\end{eqnarray}
where
\begin{eqnarray}
\left[(0,t,0)(0,0,h)(0,-t,0)\right]\circ (x,y,0)
&=&(0,t-\psi(h)(t),h)\circ (x,y,0)\nonumber\\
&=&(x,t-\psi(h)(t)+\psi(h)(y),0),
\label{f5}
\end{eqnarray}
and
\begin{eqnarray}
\left[(0,t,0)(0,0,h')(0,-t,0)\right]\circ (x',y',0)
&=&(0,t-\psi(h')(t),h')\circ (x',y',0)\nonumber\\
&=&(x',t-\psi(h')(t)+\psi(h')(y'),0).
\label{f6}
\end{eqnarray}
From Eq.~(\ref{f4}) together with Eqs.~(\ref{f5}) and (\ref{f6}), we note that $f(x,y) = f(x',y')$  if and only if
\begin{equation}
x=x'\pmod{q^{s}},~-\psi(h)(t)+\psi(h)(y)=-\psi(h')(t)+\psi(h')(y')\pmod{p^{n}},
\end{equation}
for some $h$, $h'$ and $t$ in $\Z_{p}$.

Now for any HSSP defined on $(\Z_{q^{s}}\times\Z_{p^{n}})\rtimes_\phi\Z_p$
with the set of closed subsets $\mathcal H'=\{\left(0,t,0\right)H\left(0,-t,0\right)| 0\leq t \leq p-1 \}$
and the oracle function
\begin{equation}
f(x,y) = f(x',y') \Longleftrightarrow H' \circ (x,y,0) = H' \circ (x',y',0),
\label{orif}
\end{equation}
we can always consider
the corresponding HSSP defined on $\Z_{p^n}\rtimes_\phi\Z_p$ with the set of closed subgroups
$\mathcal H=\{\left(t,0\right)H\left(-t,0\right)| 0\leq t \leq p-1 \}$ and the reduced oracle function
$g$ defined on $\Z_{p^n}\rtimes\{0\}$ such that
\begin{equation}
g(y)=g(y') \Longleftrightarrow H' \circ (y,0) = H' \circ (y',0)
\label{redf}
\end{equation}
for some $H' \in \mathcal H$.

Furthermore, to find $H'$ in $\mathcal H'$ satisfying Eq.~(\ref{orif})
for any $(x,y,0)$ and $(x',y',0)$ in $K=\Z_{q^s}\rtimes\Z_{p^n}\rtimes\{0\}$,
it is enough to find $H'$ in $\mathcal H$ satisfying Eq.~(\ref{redf})
for any $(y,0)$ and $(y',0)$ in $K=\Z_{p^n}\rtimes\{0\}$ due to the one-to-one correspondence between $\mathcal H$ and $\mathcal H'$
with respect to the oracle functions $f$ and $g$ respectively.
Thus we have the following theorem,
which states a natural reduction of HSSP on $(\Z_{q^{s}}\times\Z_{p^{n}})\rtimes_\phi\Z_p$
to HSSP on $\Z_{p^n}\rtimes_\phi\Z_p$.

\begin{Thm}
Any HSSP defined on $(\Z_{q^{s}}\times\Z_{p^{n}})\rtimes_\phi\Z_p$
with respect to the group action in Eq.~(\ref{action2})
can be naturally reduced to the HSSP on $\Z_{p^n}\rtimes_\phi\Z_p$
with respect to the group action in Eq.~(\ref{action}).
\label{red1}
\end{Thm}

Now let us consider a possible reduction of HSSP on $\Z_{N}\rtimes_\phi\Z_p$
to HSSP on $\Z_{p^n}\rtimes_\phi\Z_p$ for more general case of $N$.
By the fundamental theorem of arithmetics, $N$ can be factorized into powers of distinct primes
$N=p_1^{r_1}p_2^{r_2}\cdots p_k^{r_k}$,
and the fundamental theorem of finitely generated abelian groups implies that
$\Z_{N}$ is isomorphic to the direct product of cyclic groups
$\Z_{p_{1}^{r_{1}}} \times \Z_{p_{2}^{r_{2}}} \times \cdots \times \Z_{p_{k}^{r_{k}}}$.
Here we consider the case when $p$ does not divide each $p_j-1$ for all $j \in \{1,2,\cdots, k \}$,
and we also assume that $p=p_i$ for some $i \in \{1,2,\cdots, k \}$
to avoid the trivial case of abelian group $\Z_{N}\times\Z_p$.

For convenience, let $i=k$ and $r_k=n\geq 2$, then we have
\begin{equation}
\Z_{N}\rtimes_{\phi}\Z_{p} \cong
(\Z_{p_{1}^{r_{1}}} \times \Z_{p_{2}^{r_{2}}} \times
\cdots  \times \Z_{p_{k-1}^{r_{k-1}}} \times \Z_{p^n}) \rtimes_{\phi}\Z_{p}.
\label{eq:isomorphic}
\end{equation}
By an analogous proof of Proposition~\ref{Lem:note2},
we also note that, for each $h\in\Z_p$,
the automorphism $\phi(h)$ on $\Z_{p_{1}^{r_{1}}} \times \Z_{p_{2}^{r_{2}}} \times \cdots  \times \Z_{p^{{n}}}$
acts trivially on each component of $\Z_{p_{j}^{r_{j}}}$ such that $p$ differs from $p_{j}$.
In other words, there exists a homomorphism $\psi$ from $\Z_p$ to $\mathrm{Aut}(\Z_{p^{n}})$
such that $\phi(h)=\cI\times \psi(h)$ for each $h\in\Z_p$
where $\cI$ is the identity map on $\Z_{p_{1}^{r_{1}}} \times\cdots \times \Z_{p_{k-1}^{r_{k-1}}}$, and
\begin{equation}
\Z_{N}\rtimes_{\phi}\Z_{p}\cong
\Z_{p_{1}^{r_{1}}}  \times\cdots \times
 \Z_{p_{k-1}^{r_{k-1}}}\times (\Z_{p^{n}} \rtimes_{\psi}\Z_{p}).
 \label{eq:main}
\end{equation}

Furthermore, $\Z_{p_{1}^{r_{1}}} \times \cdots \times\Z_{p_{k-1}^{r_{k-1}}} $ is
a cyclic group of order $N/p^{n}$, thus we have
\begin{equation}
\Z_{N}\rtimes_{\phi}\Z_{p}\cong
\Z_{N/p^{n}}\times (\Z_{p^{n}} \rtimes_{\psi}\Z_{p}).
\end{equation}
Eq.~(\ref{eq:main}) implies that
solving HSSP on $\Z_{N}\rtimes_{\phi}\Z_{p}$ is essentially equivalent to
solving HSSP on $\Z_{N/p^{n}}\times (\Z_{p^{n}} \rtimes_{\psi}\Z_{p})$
because two groups are isomorphic.

Now let us consider the subgroups $K=\Z_{N/p^{n}}\times \Z_{p^{n}}\times \{0\}$,
 $H=\{(0,0)\}\times\Z_p$ of $G=\Z_{N/p^{n}}\times (\Z_{p^{n}}\rtimes_{\phi}\Z_{p})$
and the group action $\circ : G \times K \rightarrow K$ defined by
\begin{eqnarray}
\left(a,b,h \right)\circ \left(x,y,0 \right)
&=&\left(a,b,h \right)\left(x,y,0 \right)\left(0,0,-h \right)\nonumber\\
&=&\left((a,b)+\phi(h)(x,y), 0 \right),
\label{action3}
\end{eqnarray}
for any  $\left(a,b,h\right) \in G$  and $\left(x,y,0\right) \in K$
(or equivalently, for any $a, x\in \Z_{N/p^{n}}$, $b, y\in \Z_{p^n}$ and $h \in \Z_p$).
By using an analogous argument of Corollaries~\ref{H_P_t2} and \ref{H_tfix2},
it is straightforward to verify the one-to-one correspondence
between the set of closed subgroups of $\Z_{N/p^{n}}\times (\Z_{p^{n}}\rtimes_{\phi}\Z_{p})$
under the group action in Eq.~(\ref{action3})
and the set of closed subgroups of $\Z_{p^{n}} \rtimes_{\psi}\Z_{p}$.
Thus we have the following theorem about a natural reduction of HSSP on $\Z_{N}\rtimes_\phi\Z_p$
to HSSP on $\Z_{p^n}\rtimes_\phi\Z_p$ for some case of $N$.
\begin{Thm}
Let $N$ be a positive integer with a prime factorization $N=p_1^{r_1}p_2^{r_2}\cdots p_k^{r_k}$
and $p$ be an odd prime such that $p$ does not divide each $p_j-1$ for all $j \in \{1,2,\cdots, k\}$.
Then any HSSP defined on $\Z_{N}\rtimes_\phi\Z_p$ with respect to the group action in Eq.~(\ref{action3})
can be naturally reduced to the HSSP on $\Z_{p^n}\rtimes_\phi\Z_p$ with respect to the group action in Eq.~(\ref{action}).
\label{red2}
\end{Thm}

Now, together with Corollary~\ref{The: QAHSSP1}, we have the following corollary,
which states the existence of a polynomial-time quantum algorithm
solving HSSP on $\Z_{N}\rtimes_\phi\Z_p$ for some case of $N$.

\begin{Cor}
\label{The: QAHSSP2}
Let $G=\Z_{N}\rtimes_\phi\Z_p$ be the semi-direct product of cyclic groups
with an odd prime $p$ and a positive integer $N$ with a prime factorization $N=p_1^{r_1}p_2^{r_2}\cdots p_k^{r_k}$
such that $p$ does not divide each $p_j-1$ for all $j \in \{1,2,\cdots, k\}$.
$K=\Z_{N}\times \{0\}$ and $H=\{0\}\times\Z_{p}$ are two subgroups of $G$
where $G$ acts on $K$ with respect to the group action in Eq.~(\ref{action3})
and $\mathcal H$ is the set of all conjugate groups of $H$ in $G$.
Then there exists a polynomial-time quantum algorithms solving HSSP on $G$.
\end{Cor}

\section{Summary}
\label{Conclusion}
We have first investigated algebraic properties of semi-direct product of cyclic groups,
and then have presented an efficient reduction scheme of HSSP on $\Z_{N}\rtimes\Z_{p}$ to its related HSP
for the case when any prime factor $q$ of $N$
satisfies the condition that $q-1$ is not divisible by $p$.
Finally, we have proposed an efficient quantum algorithm for HSSP on the group
by applying this reduction scheme to
an efficient quantum computational algorithm for the related HSP.

\section*{Acknowledgments}
This work was supported by Emerging Technology R\&D Center of SK Telecom.
JSK was supported by
Basic Science Research Program through the National Research Foundation of Korea (NRF)
funded by the Ministry of Education, Science and Technology (2012R1A1A1012246),
and SL was supported by Kyung Hee University Research Fund in 2012.

\end{document}